\documentclass[letterpaper,english]{llncs}
\usepackage[T1]{fontenc}
\usepackage[latin9]{inputenc}
\usepackage{verbatim}
\usepackage{float}
\usepackage{url}
\usepackage{amstext}
\usepackage{graphicx}
\usepackage{enumerate}

\makeatletter

\pdfpageheight\paperheight
\pdfpagewidth\paperwidth

\floatstyle{ruled}
\newfloat{algorithm}{tbp}{loa}
\floatname{algorithm}{Algorithm}

\usepackage{hyperref}
\hypersetup{
 colorlinks=true,linkcolor=blue,anchorcolor=blue,citecolor=blue,filecolor=blue,urlcolor=blue,bookmarksnumbered=true,pdfview=FitB}
\usepackage{breakurl}

\usepackage{times} 
\usepackage[leftcaption]{sidecap}
\usepackage{wrapfig} 

\@ifundefined{showcaptionsetup}{}{%
 \PassOptionsToPackage{caption=false}{subfig}}\usepackage{subfig}
\makeatother

\usepackage{babel}

\newcommand{\Leaves}{\mathcal{L}}
\newcommand{\Info}{\mathcal{I}}

\begin{document}

\global\long\def\searchTree{\textrm{SearchTree}}
\global\long\def\mulT{\mbox{MUL-tree}}
\global\long\def\mT{M\mbox{-tree}}

\title{Extracting Conflict-free Information from Multi-labeled Trees\thanks{This work was supported in part by the National Science Foundation under grant DEB-0829674.}}

\author{Akshay Deepak\inst{1} 
\and David
Fern\'andez-Baca\inst{1} \and Michelle M. McMahon\inst{2}}

\institute{Department of Computer Science, Iowa State University, Ames, IA 50011, USA \and
School of Plant Sciences, University of Arizona, Tucson, AZ 85721, USA}
\maketitle

\begin{abstract}
A multi-labeled tree, or $\mulT$, is a phylogenetic tree where two or more leaves share a label, e.g., a species name.  A \mulT\ can imply multiple conflicting phylogenetic relationships for the same set of taxa, but can also contain conflict-free information that is of interest and yet is not obvious. We define the information content of a \mulT\ $T$ as the set of all conflict-free quartet topologies implied by $T$, and define the maximal reduced form of $T$ as the smallest tree that can be obtained from $T$ by pruning leaves and contracting edges while retaining the same information content.  We show
that any two $\mulT$s with the same information content
exhibit the same reduced form. This introduces an equivalence relation
in $\mulT$s with potential applications to comparing $\mulT$s. 
We present an efficient algorithm to reduce a \mulT\ to its maximally reduced form and evaluate its performance on empirical datasets in terms of both quality of the reduced tree and the degree of data reduction achieved.  %
\end{abstract}

\section{Introduction}

Multi-labeled trees, also known as $\mulT$s, are phylogenetic trees that can have
more than one leaf with the same label \cite{Fellows2003192,Grundt2004695,huber2006phylogenetic,Popp2001474,scornavacca2010building} (Fig.~\ref{Flo:Mul-treeEg}).  $\mulT$s arise naturally and frequently in data sets containing multiple
genes or gene sequences for the same species~\cite{sanderson2008phylota}, but they can also arise in
bio-geographical studies or co-speciation studies where leaves represent individual taxa yet are labeled with their areas \cite{ganapathy2006pattern} or hosts \cite{Johnson_et_al_2003}.

MUL-trees, unlike singly-labeled trees, can contain conflicting species-level phylogenetic information due, e.g., to whole genome duplications~\cite{lott2009inferring}, incomplete lineage sorting~\cite{RasmussenKellis2012}, inferential error, or, frequently, an unknown combination of several factors. However, they can also contain substantial amounts of conflict-free information.
Here we provide a way to extract this information; specifically, we have the following results.
\begin{itemize}
\item
We introduce a new quartet-based measure of the information content of a \mulT, defined as the set of conflict-free quartets the tree displays (Section~\ref{sec:Preliminaries-and-Definitions}).  
\item
We introduce the concept of the maximally-reduced form (MRF) of a \mulT, the smallest \mulT\ with the same information content (Section~\ref{sec:MRF}), and
show that any two $\mulT$s with the same information content have the same MRF (Theorem~\ref{thm:10}).  
\item
We present a simple algorithm to construct the MRF of a $\mulT$ (Section~\ref{sec:Algorithm}); its running time is quadratic in the number of leaves and does not depend on the multiplicity of the leaf labels or the degrees of the internal nodes.  
\item 
We present computational experience with an implementation of our MRF algorithm (Section~\ref{sec:Application}).   In our test data, the MRF is often significantly smaller than the original tree, while retaining most of the taxa.  
\end{itemize}

We now give the intuition behind our notion of  information content, deferring the formal definitions of this and other concepts to the next section.
Quartets (i.e., sets of four species) are a natural starting point, since they are the smallest subsets from which we can draw meaningful topological information.  A singly-labeled tree implies exactly one topology on any quartet.  More precisely, each edge $e$ in a singly-labeled tree implies a bipartition $(A,B)$ of the leaf set, where each part is the set of leaves on one the two sides of $e$.  From $(A,B)$, we derive a collection of bipartitions $ab|cd$ of quartets, such that $\{a,b\} \subseteq A$ and $\{c,d\} \subseteq B$.  Clearly, if one edge in a singly-labeled tree implies some bipartition $q= ab|cd$ of $\{a,b,c,d\}$, then there can be no other edge that implies a bipartition, such as $ac|bd$, that is in conflict with $q$.    Indeed, the quartet topologies implied by a singly-labeled tree uniquely identify it \cite{steel1992complexity}.

\begin{wrapfigure}{r}{0.44\textwidth}
\vspace{-30pt}
\begin{center}
\includegraphics[scale=0.55]{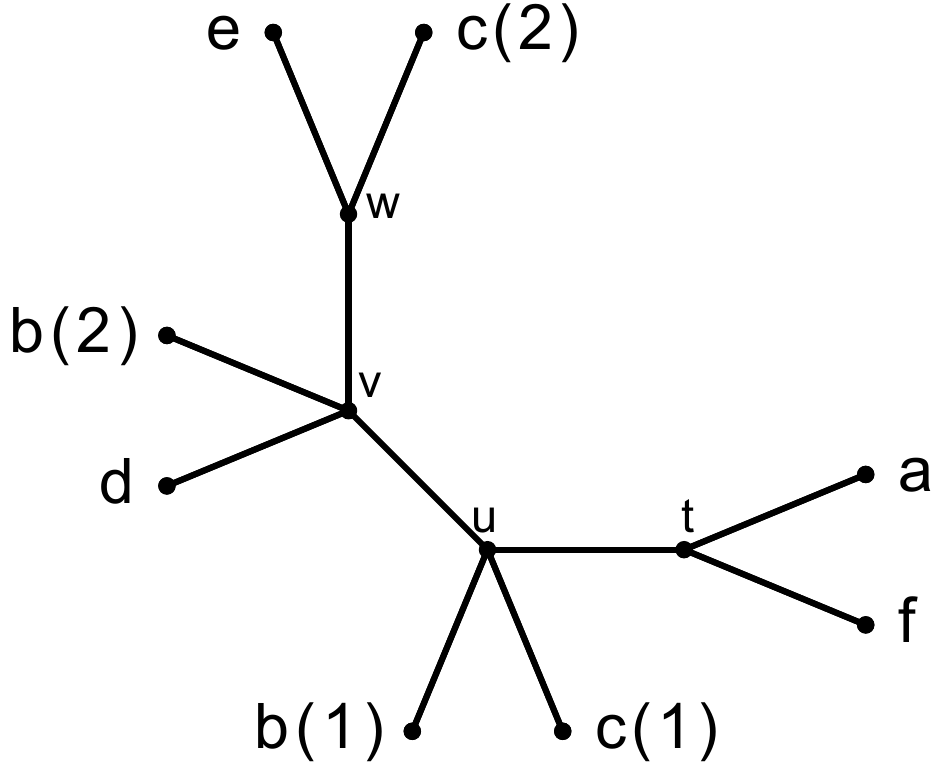}
\end{center}
\vspace{-20pt}
\caption{A $\mulT$. 
Numbers in  parenthesis next to labels indicate the multiplicity of the respective labels
and are not part of the labels themselves.}
\vspace{-20pt}
\label{Flo:Mul-treeEg}
\end{wrapfigure}

The situation for $\mulT$s is more complicated, as illustrated in Fig.~\ref{Flo:Mul-treeEg}.  
Here, the presence of two copies of labels  $b$ and $c$ --- $b(1)$ and $b(2)$, and $c(1)$ and $c(2)$ --- leads to two conflicting topologies on the quartet $\{b,c,d,e\}$. Edge $(u,v)$, implies the bipartition $bc|de$, corresponding to the labels $\{b(1),c(1),d,e\}$, while edge $(v,w)$ implies $bd|ce$ corresponding to the leaves $\{b(2),c(2),d,e\}$.  On the other hand, the quartet topology $af|bc$, implied by edge $(t,u)$, has no conflict with any other topology that the tree exhibits on $\{a,b,c,f\}$.  We show that the set of all such conflict-free quartet topologies is compatible (Theorem~\ref{thm:singly_labeled}).  That is, for every $\mulT$ $T$ there exists at least one singly-labeled tree that displays all the conflict-free quartets of $T$ --- and possibly some other quartets as well.  
Motivated by this, we only view  conflict-free quartet topologies as informative, and define the information content of a $\mulT$ as the set of all conflict-free quartet topologies it implies. 

Conflicting quartets may well provide information, whether about paralogy, deep coalescence, or mistaken annotations. In some cases, species-level phylogenetic information can be recovered from conflicted quartets through application of, e.g., gene-tree species-tree reconciliation, an NP-hard problem. However, this is not feasible when the underlying cause of multiplicity is unknown and when conducting large-scale analyses. Our definition of information content specifically allows us to remain agnostic with respect to the cause and conservative with respect to species relationships, i.e., it does not introduce quartets not originally supported by the data.


\begin{wrapfigure}{r}{0.36\textwidth}
\begin{center}
\includegraphics[scale=0.55]{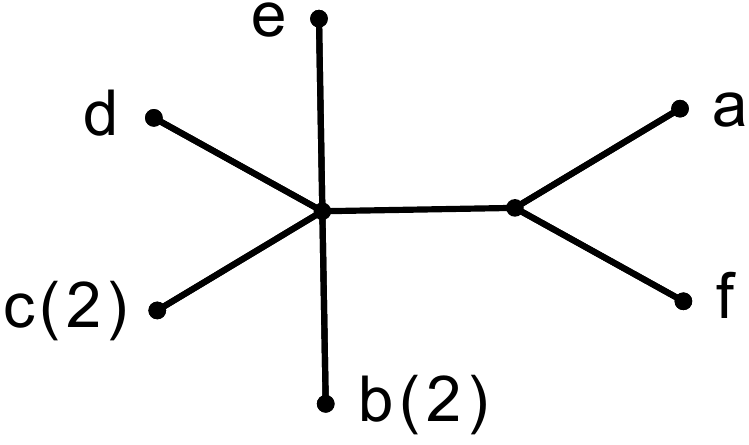}
\end{center}
\vspace{-20pt}
\caption{The MRF for the Mul-tree in Fig.~\ref{Flo:Mul-treeEg}.}
\vspace{-20pt}
\label{Flo:Mul-treeEgMRF}
\end{wrapfigure}

A \mulT\ may have leaves that can be pruned and edges that can be contracted without altering the tree's information content, i.e., without adding or removing conflict-free quartets.  For example, in Fig.~\ref{Flo:Mul-treeEg}, every quartet topology that edge $(v,w)$ implies is either in conflict with some other topology (e.g., for set $\{b,c,d,e\}$) or is already implied by some other edge (e.g., $af|ce$ is also implied by $(t,u)$).  Thus, $(v,w)$ can be contracted without altering the information content.  In fact, the information content remains unchanged if we also contract $(u,v)$ and remove the leaves labeled $b(1)$ and $c(1)$.  We define the MRF of a \mulT\ $T$ as the tree that results from applying information-preserving edge contraction and leaf pruning operations repeatedly to $T$, until it is no longer possible to do so. For the tree in Fig.~\ref{Flo:Mul-treeEg}, the MRF is singly-labeled as shown in Fig.~\ref{Flo:Mul-treeEgMRF}. Note that, in general, the MRF may not be singly-labeled (see the example in Section~\ref{sec:Example}). If the MRF is itself a $\mulT$, it is not possible to reduce the original to a singly-labeled tree without either adding at least one quartet that did not exist conflict-free in $T$ or by losing one or more conflict-free quartets.

Since any two $\mulT$s with the same information content have the same MRF, rather than comparing $\mulT$s directly, we can instead compare their MRFs.   
This is appealing mathematically, because it focuses on conflict-free information content, and also computationally, since an MRF can be much smaller than the original $\mulT$. 
Indeed, on our test data, the MRF was frequently singly-labeled.   This reduction in input size is especially significant if the \mulT\ is an input to an algorithm whose running time is exponential in the label multiplicity, such as Ganapathy et al.'s algorithm to compute the contract-and-refine distance between two area cladograms \cite{ganapathy2006pattern} or Huber et al.'s algorithm to determine if a collection of ``multi-splits'' can be displayed by a MUL-tree \cite{huber2008complexity}. 

For our experiments, we also implemented a post-processing step, which converts the MRF to a singly-labeled tree, rendering it available for  analyses that require singly-labeled trees, including supermatrix \cite{de2007supermatrix,Wiens1995Combining} and supertree methods  \cite{Baum1992Combining,ragan1992phylogenetic,BansalBEFB2010,swenson2012superfine}.  
On the trees in our data set, the combined taxon loss between the MRF computation and the postprocessing was much lower than it would have been had we simply removed all duplicate taxa from the original trees.

Previous work on $\mulT$s has concentrated on finding ways to reduce $\mulT$s to singly-labeled trees (typically in order to provide inputs to supertree methods) 
\cite{scornavacca2010building}, and to develop metrics and algorithms to compare $\mulT$s \cite{ganapathy2006pattern,puigbo2007topd,marcet2011treeko,huber2010metrics}.  In contrast to our approach --- which is purely topology-based and is agnostic with respect to the cause of label multiplicity ---, the assumption underlying much of the literature on $\mulT$s is that taxon multiplicity results from gene duplication. 
Thus, methods to obtain singly-labeled trees from $\mulT$s usually work by pruning subtrees at putative duplication
nodes. Although the proposed algorithms are polynomial, they are unsatisfactory in various ways. For example, in \cite{scornavacca2010building}
if the subtrees are neither identical nor compatible, then the subtree
with smaller information content is pruned, which seems to discard too much information. Further, the algorithm is only efficient for binary rooted trees. In \cite{puigbo2007topd}
subtrees are pruned arbitrarily, while in \cite{marcet2011treeko} 
at each putative duplication node a separate analysis is done for each possible pruned subtree. 
Although the latter approach is better than pruning arbitrarily, in the worst case it can end up analyzing exponentially many subtrees.

%
{}

%
{}

\section{\label{sec:Preliminaries-and-Definitions}MUL-Trees and Information Content}

A \emph{$\mulT$} is a triple
$(T,M,\psi)$, where (i) $T$ is an unrooted tree\footnote{The results presented here can be extended to rooted trees, using triplets instead of quartets, exploiting the well-known bijection between rooted and unrooted trees~\cite[p.~20]{sempleSteelPhylogenetics}.  We do not discuss this further here for lack of space.} with leaf set $\Leaves(T)$ all of whose internal nodes have degree at least three, (ii) $M$ is a set of labels,
and (iii) $\psi_{T}:\Leaves(T) \rightarrow M$ is a surjective map that
assigns each leaf of $T$ a label from $M$.  (Note that if $\psi$ is a bijection, $T$ is singly labeled; that is, singly-labeled trees are a special case of $\mulT$s.)
For brevity we often refer to a $\mulT$ by its underlying tree $T$.  In what follows, unless stated otherwise,
by a tree we mean a $\mulT$. 

An edge $(u,v)$ in $T$ is \emph{internal}
if neither $u$ nor $v$ belong to $\Leaves(T)$, and is \emph{pendant} otherwise. A \emph{pendant node} is an internal node that has a leaf as its neighbor.

Let $(u,v)$ be an edge in $T$ and $T'$ be the result of deleting $(u,v)$ from $T$.  Then $T_{u}^{uv}$ ($T_{v}^{uv}$) denotes the subtree of $T'$ that contains
$u$ ($v$). $M_{u}^{uv}$ ($M_{v}^{uv}$) denotes the set
of labels in $T_{u}^{uv}$ ($T_{v}^{uv}$) but not in $T_{v}^{uv}$ ($T_{u}^{uv}$). $C^{uv}$ is the set of labels common to both $T_{u}^{uv}$
and $T_{v}^{uv}$.  Observe that $M_{u}^{uv}$, $M_{v}^{uv}$ and $C^{uv}$ partition $M$.
 For example, in Fig.~\ref{Flo:Mul-treeEg},
$M_{u}^{uv}=\left\{ a,f\right\} $, $M_{v}^{uv}=\left\{ e,d\right\} $,
$C^{uv}=\left\{ b,c\right\} $. 

A (resolved) \emph{quartet} in a $\mulT$ $T$ is a bipartition $ab|cd$ of a set of labels $\left\{ a,b,c,d\right\}$ such that there is an edge $(u,v)$ in $T$ with $\left\{ a,b\right\} \in M_{u}^{uv}$
and $\left\{ c,d\right\} \in M_{v}^{uv}$. We say that $(u,v)$ \emph{resolves} $ab|cd$.
For example, in Fig.~\ref{Flo:Mul-treeEg}, edge $(t,u)$
resolves $af|bc$.  

%
%


The \emph{information content of an edge}
$(u,v)$ of a \mulT\ $T$, denoted $\Delta(u,v)$, is the set of quartets
resolved by $(u,v)$. 
An edge $(u,v)$ in tree $T$ is \emph{informative}
if $|\Delta(u,v)|>0$;  $(u,v)$ is 
\emph{maximally informative} if there is no other edge
$(u',v')$ in $T$ with $\Delta(u,v)\subset\Delta(u',v')$.
The \emph{information content} of
$T$, denoted $\Info(T)$, is the combined information content of all edges in the tree; that is $\Info(T)=\bigcup_{(u,v)\text{\ensuremath{\in}}E}\Delta(u,v)$, where $E$ denotes the set of edges in $T$.
 
The next result shows that the quartets in $\Info(T)$ are conflict-free.

\begin{theorem}
\label{thm:singly_labeled}
For every  \mulT\ $T$, there is a singly labeled tree $T'$ such that $\Info(T) \subseteq \Info(T')$.\footnote{All proofs are in the Appendix.  The Appendix will not be part of the final submission.
}
\end{theorem}

Note that  there are examples where the containment  indicated by the above result is proper.

To conclude this section, we give some results that are useful for the reduction algorithm of Section~\ref{sec:Algorithm}.
In the next lemmas, $(u,v)$ and $(w,x)$ denote
two edges in tree $T$ that lie on the path $P_{u,x} = (u,v, \dots , w,x)$
as shown in Fig.~\ref{Flo:ProofLemma1}.  

\begin{wrapfigure}{r}{0.40\textwidth}
\begin{center}
\includegraphics[width=0.40\textwidth]{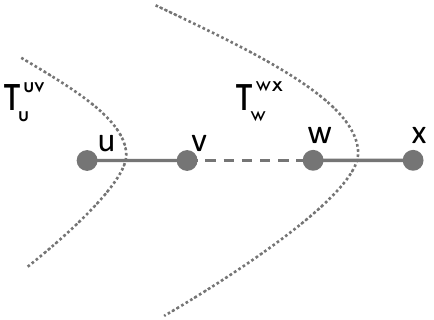}
\end{center}
\vspace{-20pt}
\caption{}
\vspace{-45pt}
\label{Flo:ProofLemma1}
\end{wrapfigure}

\begin{lemma}
\label{lem:2}If $|M_{u}^{uv}|=|M_{w}^{wx}|$ then $M_{u}^{uv}=M_{w}^{wx}$.  Otherwise, $M_{u}^{uv}\subset M_{w}^{wx}$.\end{lemma}

Together with Lemma \ref{lem:2}, the next result allows us to check whether the information content of an edge is a subset of that of another based solely on the
cardinalities of the $M_{u}^{uv}$s.

\begin{lemma}
\label{lem:3}$\Delta(u,v)\subseteq\Delta(w,x)$ if and only if $M_{v}^{uv}=M_{x}^{wx}$. 
\end{lemma}

\begin{lemma}
\label{lem:4}Suppose $\Delta(u,v)\subseteq\Delta(w,x)$.  Then, for any edge $(y,z)$ on $P_{u,x}$ such that $v$ is closer to $y$ than to $z$, $\Delta(u,v)\subseteq\Delta(y,z)\subseteq\Delta(w,x)$. 
\end{lemma}

\section{Maximally Reduced MUL-Trees}
\label{sec:MRF}

Our goal is to provide a way to reduce a \mulT\ $T$ as much as possible, while preserving its information content.  Our reduction algorithm uses the following operations.
\begin{description}
\item[Prune$(v)$:]  Delete leaf  $v$ from $T$.  If, as a result, $v$'s neighbor $u$ becomes a degree-two node, connect the former two neighbors of $u$ by an edge and delete $u$. 
\item[Contract$(e)$:] Delete an internal edge $e$ and identify its endpoints.
\end{description}

A leaf $v$ in $T$ is \emph{prunable} if the tree that results from pruning  $v$ has the same information content as $T$. An internal edge $e$ in $T$ is \emph{contractible} if the tree that results from contracting $e$ has the same information content as $T$.
$T$ is \emph{maximally
reduced} if it has no prunable leaf and no contractible internal edge.

\begin{theorem}
\label{thm:5}Every internal edge in a maximally reduced tree $T$ resolves
a quartet that is resolved by no other edge. 
\end{theorem}

Note that a quartet $ab|cd$ that is resolved by edge $(u,v)$, but by no other edge  must have the form illustrated in Fig.~\ref{Flo:uniqueQuartet}.

\begin{wrapfigure}{r}{0.35\textwidth}
\vspace{-30pt}
\begin{center}
\includegraphics[width=0.35\textwidth]{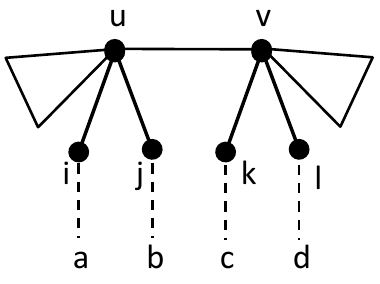}
\end{center}
\vspace{-20pt}
\caption{Quartet $ab|cd$ is resolved only by edge $(u,v)$. Here, $a\in M_{i}^{ui}$, $b\in M_{j}^{uj}$, $c\in M_{k}^{vk}$ and $d\in M_{l}^{vl}$.}
\vspace{-30pt}
\label{Flo:uniqueQuartet}
\end{wrapfigure}

Next, we show that the set of quartets resolved by a maximally reduced
tree uniquely identifies the tree. 

\begin{theorem}
\label{thm:10}Let $T$ and $T'$ be two maximally reduced trees such that $\Info(T)=\Info(T')$.  Then, $T$ and $T'$ are isomorphic.
\end{theorem}

The \emph{maximally reduced form} (MRF) of a \mulT\ $T$ is the tree that results from repeatedly pruning prunable leaves and contracting contractible edges from $T$ until this is no longer possible.  Theorem~\ref{thm:10} shows that we can indeed talk about ``the'' MRF of $T$.

\begin{corollary}
Every $\mulT$ has a unique MRF.
\end{corollary}

\begin{corollary}
Any two $\mulT$s with the same information content have the same
MRF.
\end{corollary}
\begin{corollary}\label{cor:3}
If a maximally reduced \mulT\ $T$  is not singly-labeled, there
does not exist a singly-labeled tree having the same information content as $T$. 
\end{corollary}
\begin{wrapfigure}{r}{0.35\textwidth}
\begin{center}
\includegraphics[width=0.35\textwidth]{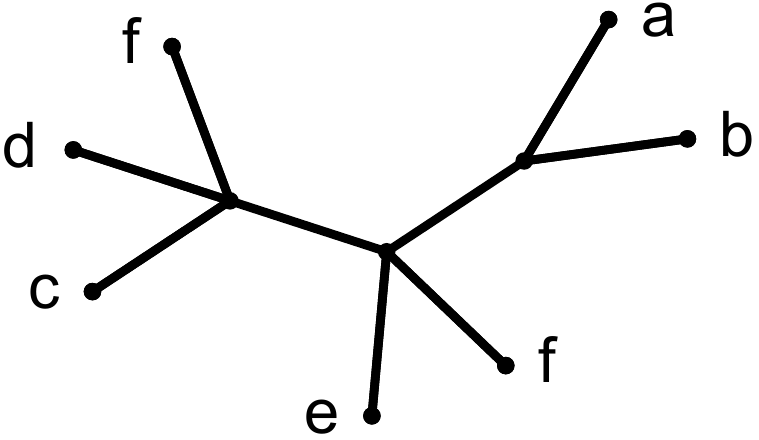}
\end{center}
\vspace{-20pt}
\caption{A maximally reduced $\mulT$}
\vspace{-20pt}
\label{Flo:irreplacable_Mul-tree}
\end{wrapfigure}

Fig \ref{Flo:irreplacable_Mul-tree} illustrates the last result. 
Any singly-labeled tree resolving the same set of quartets must be
obtained by removing one of the leaves labeled with $f$. However,
doing so will also introduce quartets that are not resolved by the
maximally reduced $\mulT$.
Note that Corollary \ref{cor:3} does not contradict with Theorem \ref{thm:singly_labeled}. If the $\mulT$ 
in Theorem \ref{thm:singly_labeled} is maximally reduced and not singly-labeled, the containment is proper;
i.e., $\Info(T) \neq \Info(T')$, which is the claim of Corollary \ref{cor:3}.

\begin{corollary}
The relation \textquotedblleft{}sharing a common MRF\textquotedblright{}
is an equivalence relation on the set of $\mulT$s .
\end{corollary}

The last result implies that $\mulT$s can be partitioned into equivalence classes, where each class consists of the set of all trees with the same information content.  Thus, instead of comparing $\mulT$s directly, we can compare their maximally reduced forms.  

To end this section, we give some results that help to identify contractible edges and prunable leaves.   The setting is the same as for Lemmas~\ref{lem:3} and~\ref{lem:4}: $(u,v)$ and $(w,x)$ are two edges in tree $T$ that lie on the path $P_{u,x} = (u,v, \dots , w,x)$
(see Fig.~\ref{Flo:ProofLemma1}).
We say that subtree $T_{z}^{yz}$ \emph{branches out} from the path $P_{u,x}$
if $y\in P_{u,x} - \{u,x\}$, and $z\notin P_{u,x}$.

\begin{lemma}
\label{lm:prune_n_contract}
Suppose $\Delta(u,v)\subseteq\Delta(w,x)$ then
\begin{enumerate}[(i)]
\item every internal edge on a subtree branching out from $P_{u,x}$ is contractible, and 
\item if $\Delta(u,v)=\Delta(w,x)$, every leaf on a subtree branching out from $P_{u,x}$ is prunable.  Thus, the entire subtree can be deleted without changing the information content of the tree. 
\end{enumerate}
\end{lemma}

\section{\label{sec:Algorithm}The Reduction Algorithm}

We now describe a $O(n^2)$ algorithm to compute the MRF of an $n$-leaf \mulT\ $T$.  In the previous section, the MRF was defined as the tree obtained by applying information-preserving pruning and contraction operations to $T$, in any order, until it is no longer possible.  For efficiency, however, the sequence in which these steps are performed is important.  Our algorithm has three distinct phases: a preprocessing step, redundant edge contraction, and pruning of redundant leaves.  We describe these next and then give an example.

\subsection{Preprocessing}

For every edge $(u,v)$ in $T$, we compute $|M_{u}^{uv}|$
and $|M_{v}^{uv}|$. This can be done in $O(n^2)$ time as
follows.  First, traverse subtrees $T_{u}^{uv}$ and $T_{v}^{uv}$ to
count number of distinct labels $n_{u}^{uv}$ and $n_{v}^{uv}$ in each subtree. Then, $|M_{u}^{uv}|=|M|- n_{v}^{uv}$ and $|M_{v}^{uv}|=|M|-n_{u}^{uv}$.  We then contract non-informative edges; i.e., edges $(u,v)$ where $|M_u^{uv}|$ or $|M_v^{uv}|$ is at most one.

\subsection{Edge Contraction and Subtree Pruning}

Next, we repeatedly find pairs of adjacent edges $(u,v)$ and $(v,w)$ such that $\Delta(u,v) \subseteq \Delta(v,w)$ or vice-versa, and contract the less informative of the two.  By Lemmas~\ref{lem:2} and \ref{lem:3}, we can compare $\Delta(u,v)$ and $\Delta(v,w)$ in constant time using the precomputed values of $|M_{u}^{uv}|$ and $|M_{v}^{uv}|$.
Lemma~\ref{lm:prune_n_contract}(i) implies that we should also contract all internal edges incident on $v$ or in the subtrees branching out of $v$.  Further, by Lemma~\ref{lm:prune_n_contract}(ii), if $\Delta(u,v) = \Delta(v,w)$, we can in fact delete these subtrees entirely, since their leaves are prunable. 
Lemma~\ref{lem:4} implies that all such
edges must lie on a path, and hence can be identified in linear
time. The total time for all these operations is linear, since at worst we traverse every edge
twice. 

\subsection{Pruning Redundant Leaves}

The tree that is left at this point has no contractible edges; however,  it can still have prunable leaves.  We first prune any leaf with a label $\ell$ that does not participate
in any resolved quartet.  Such an $\ell$ has the property that for every edge $(u,v)$, $\ell\notin M_{u}^{uv}$ and $\ell\notin M_{v}^{uv}$. 
All such leaves can be found in $O(n^{2})$ time and
$O(n)$ space.  

Next, we consider sets of leaves with the same label $\ell$ that share a common neighboring pendant node.  Such leaves can be found in linear time.  For each such set, we delete all but one element. 

After such leaves are removed, let $T$ be the resulting tree. Now, the only kind of prunable leaf with a given label $\ell$ that might remain are leaves attached to different pendant
nodes. The next result identifies such redundant leaves in $T$.

\begin{lemma}
\label{lm:prunable}
Let $\ell$ be a multiply-occurring label in $T$ and
let $T'$ be the minimal subtree that spans all the leaves labelled by $\ell$.  Then, any leaf in $T$ labeled $\ell$ attached to a
pendant node in $T'$ of degree at least three is prunable.
\end{lemma}

Thus, to identify and prune redundant leaf
nodes of the latter type:
\begin{enumerate}
\item For each label $\ell$, consider the subgraph on the leaves labeled
by it.
\item In this subgraph, delete any leaf not attached to a degree 2
pendant node as it is a redundant leaf.
\end{enumerate}

This takes $O(n)$ time per label and $O(n^{2})$
time total.  The space used is $O(n)$. Hence, the overall time and space complexities are $O(n^{2})$ time and $O(n)$, respectively.

The resulting tree has no contractible edges nor prunable leaves.  Therefore, it is the MRF of the orginal \mulT.

\subsection{\label{sec:Example}An Example}

We illustrate the reduction of the unrooted $\mulT$ shown in Fig.~\ref{Flo:appen-eg-input} to its MRF.

%
\begin{figure}
\begin{centering}
\includegraphics[scale=0.43]{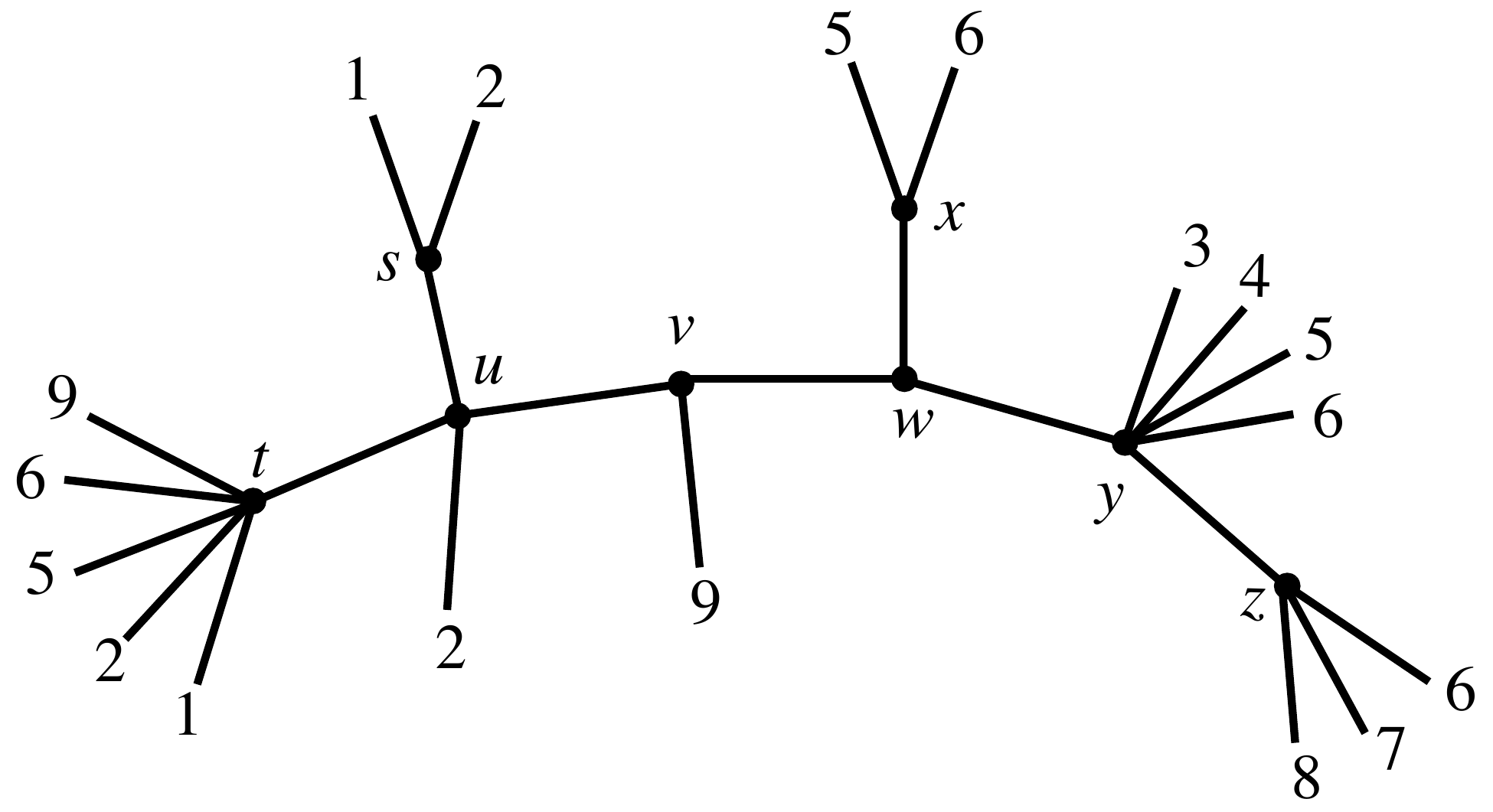}
\par\end{centering}
\vspace{-10pt}

\caption{}

\label{Flo:appen-eg-input}%
\end{figure}

\begin{enumerate}
\item
In the preprocessing step, we find that $M_t^{t u} = \emptyset$, $M_s^{s u} = \emptyset$ and $M_x^{w x} = \emptyset$, so edges $(t,u)$, $(s,u)$ and $(w,x)$ are uninformative.  They are therefore contracted, resulting in the tree shown in Fig.~\ref{Flo:preprocess_step}.

\begin{figure}
\begin{centering}
\includegraphics[scale=0.43]{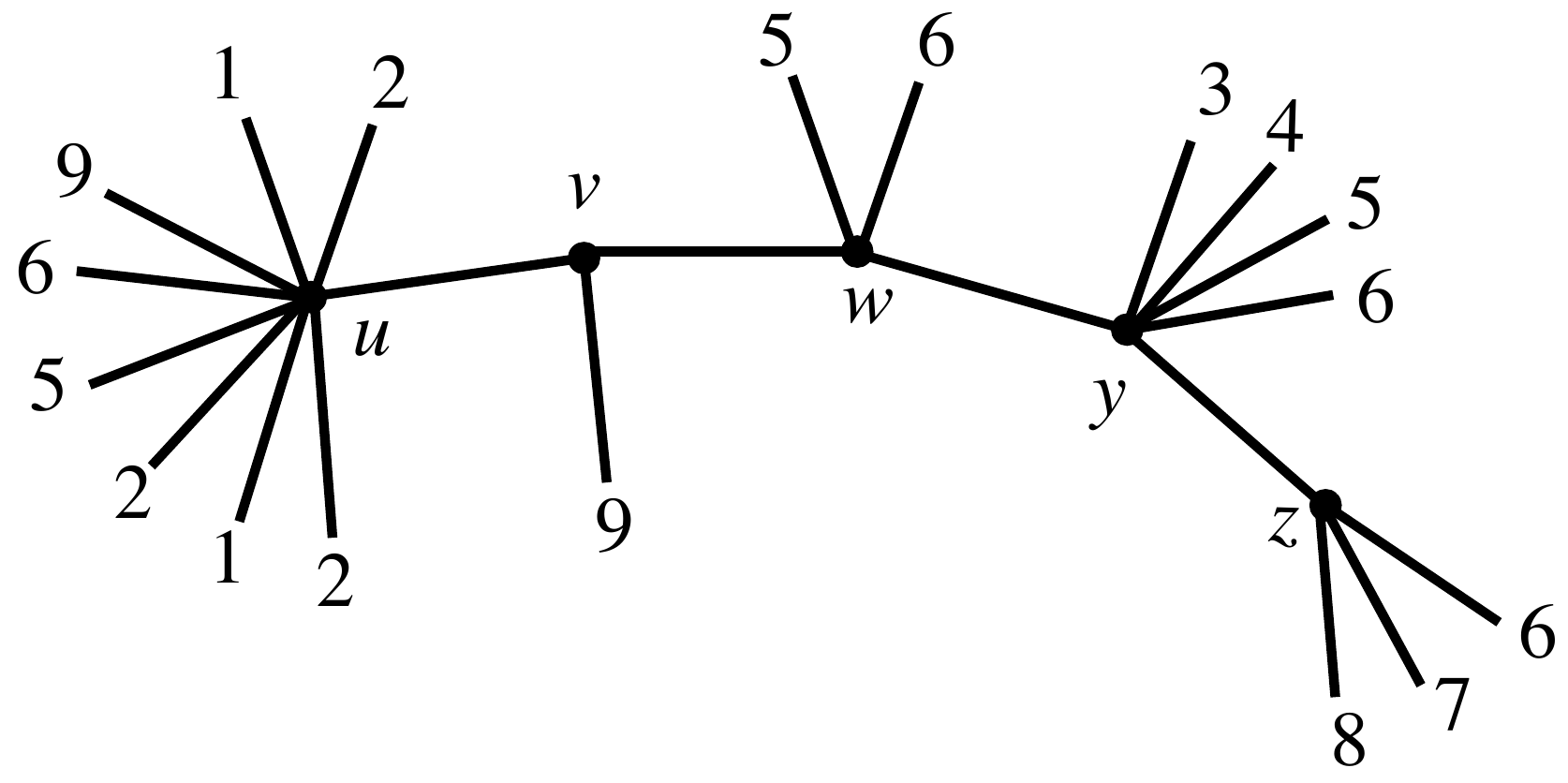}
\par\end{centering}
\vspace{-10pt}

\caption{}

\label{Flo:preprocess_step}%
\end{figure}

\item Since $\Delta(u,v) \subset \Delta(v,w)$, contract $(u,v)$. The
result is shown in Fig.~\ref{Flo:contract_1}.

\begin{figure}
\begin{centering}
\includegraphics[scale=0.43]{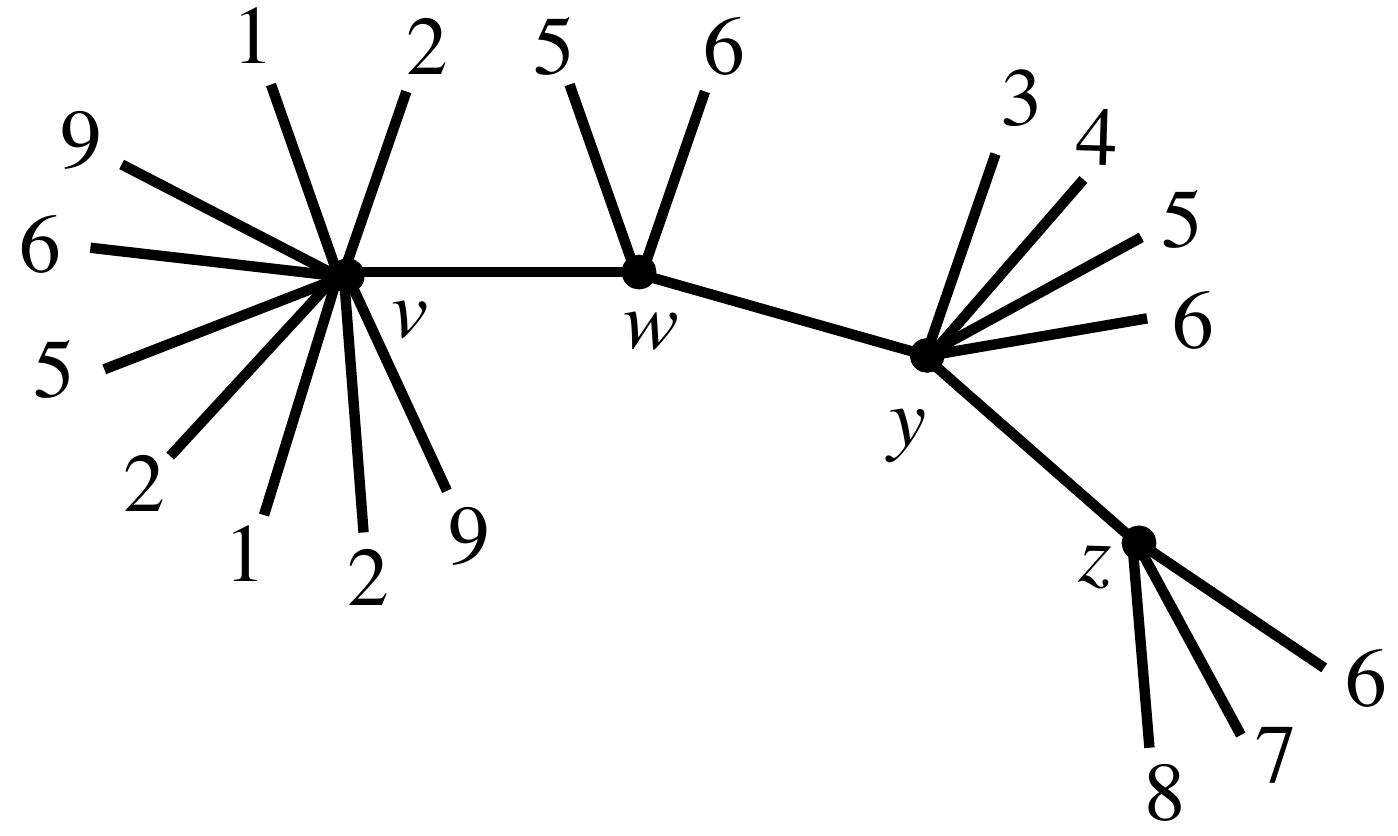}
\par\end{centering}
\vspace{-10pt}

\caption{}

\label{Flo:contract_1}%
\end{figure}
\item Since $\Delta(v,w)=\Delta(w,y)$, delete the subtree branching out at $w$ from the path from $v$ to $y$ and contract $(v,w)$. The
result is shown in Fig.~\ref{Flo:contract_2}.
\begin{figure}
\begin{centering}
\includegraphics[scale=0.43]{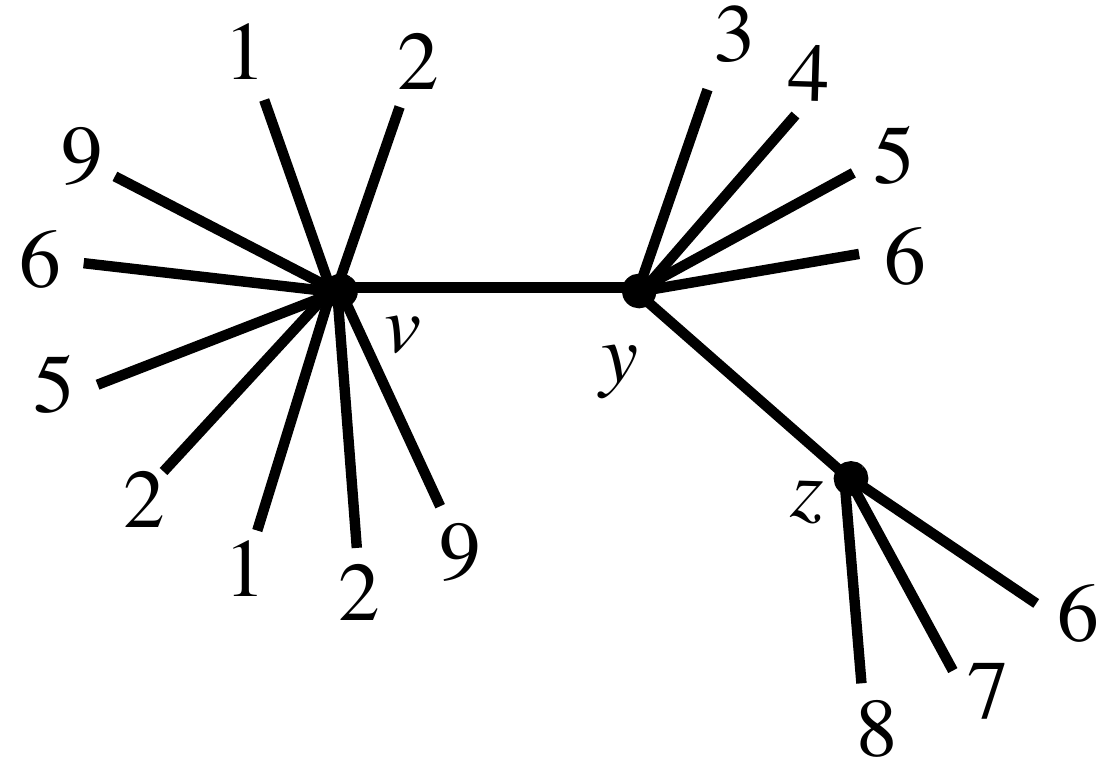}
\par\end{centering}
\vspace{-10pt}

\caption{}

\label{Flo:contract_2}%
\end{figure}

\item Prune taxon $6$, which does not participate in any quartet, and all duplicate taxa at the pendant nodes. The result, shown in Fig.~\ref{Flo:final}, is the MRF of the original tree.
\begin{figure}
\begin{centering}
\includegraphics[width=0.4\linewidth]{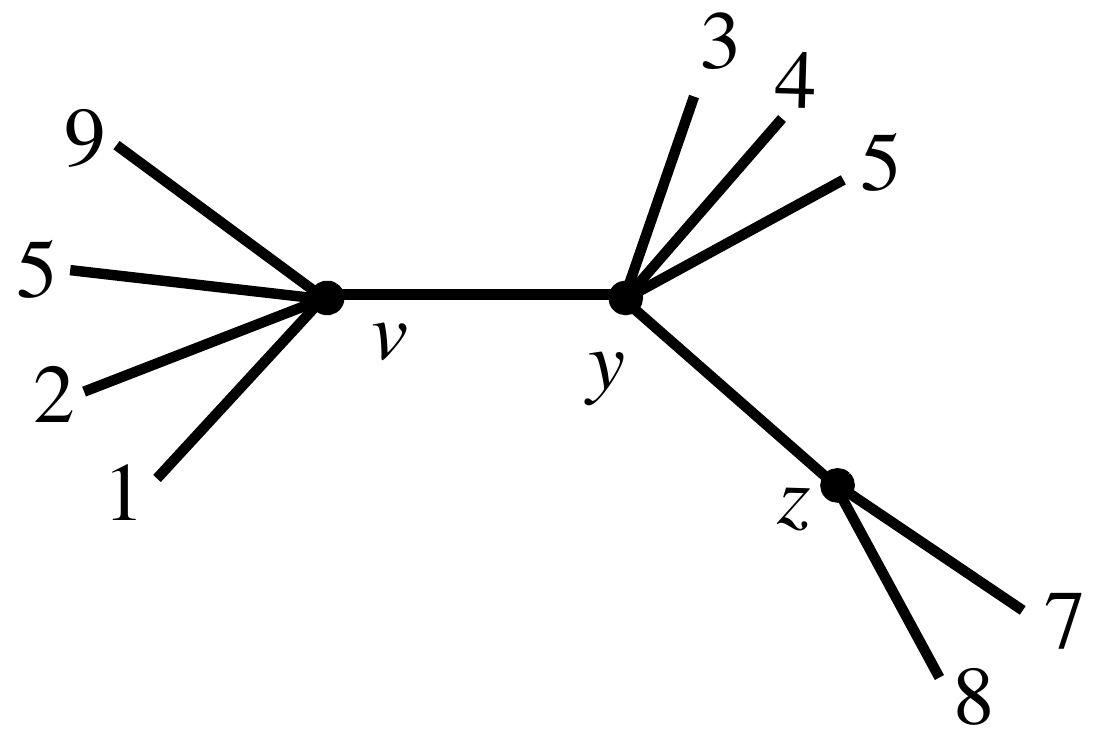}
\par\end{centering}
\vspace{-10pt}

\caption{}

\label{Flo:final}%
\end{figure}

\end{enumerate}

\section{\label{sec:Application}Evaluation}

We implemented our \mulT\ reduction algorithm, as well as a second step that restricts the MRF to the set of labels that appear only once, which yields a singly-labeled
tree.  We tested our two-step program on a set of 110,842
$\mulT$s obtained from the PhyLoTA database [10] (\url{http://phylota.net/}; GenBank eukaryotic nucleotide sequences, release 184, June 2011), which included a broad range of label-set sizes, from 4 to 1500 taxa.

There were 8,741 trees (7.8\%) with essentially no information content; these lost all resolution either when reduced to their MRFs, or in the second step. 
The remaining trees fell into two categories.  Trees in set $A$ had a singly-labeled MRF; 
65,709 trees (59.3\%) were of this kind. Trees in set $B$ were reduced to singly-labeled trees in the second step; 36,392 trees (32.8\%) were of this kind. Reducing a tree to its MRF (step 1),  led to an average taxon loss of 0.83\% of the taxa in the input $\mulT$. The total taxon loss after the second step (reducing the MRFs in set $B$ to singly-labeled trees), averaged 12.81\%. This taxon loss is not trivial, but it is far less than the 41.27\% average loss from the alternative, na\"ive, approach in which all mul-taxa (taxa that label more than one leaf) are removed at the outset. 
Note that, by the definition of MRFs, taxa removed in the first step do not contribute to the information content, since all non-conflicting quartets are preserved. On the other hand, taxa removed in the second step do alter the information content, since each such taxon participated in some non-conflicting quartet. 

Taxon loss is sensitive to the number of total taxa and, especially, mul-taxa, as demonstrated in Figure \ref{Flo:TaxaLoss}. The grey function shows
the percentage of mul-taxa in the original input trees, which is the taxon loss if we had restricted the input $\mulT$s to the set of singly-labeled leaves. The black function shows the percentage of mul-taxa lost after steps 1 and 2 of our reduction procedure. 

In addition to the issue of taxon loss, we investigated the effect of our reduction on edge loss, i.e., the level of resolution within the resulting singly-labeled tree. Input  $\mulT$s were binary and therefore had more nodes than twice the number of taxa (Fig~\ref{Flo:reducedTreeQuality}, solid line), whereas a binary tree on singly labeled taxa would have approximately as many nodes as twice the number of taxa (Fig.~\ref{Flo:reducedTreeQuality}, dashed line). We found that, although there was some edge loss, the number of nodes in the reduced singly-labeled trees (Fig.~\ref{Flo:reducedTreeQuality}, dotted line) corresponded well to the total possible, indicating low levels of edge loss. 
Note that each point on the dotted or solid lines represents an average over all trees with the same number of taxa. 

\begin{figure}[t]
\begin{centering}
\subfloat[Taxon loss in the second step]{\begin{centering}
\includegraphics[width=0.82\columnwidth]{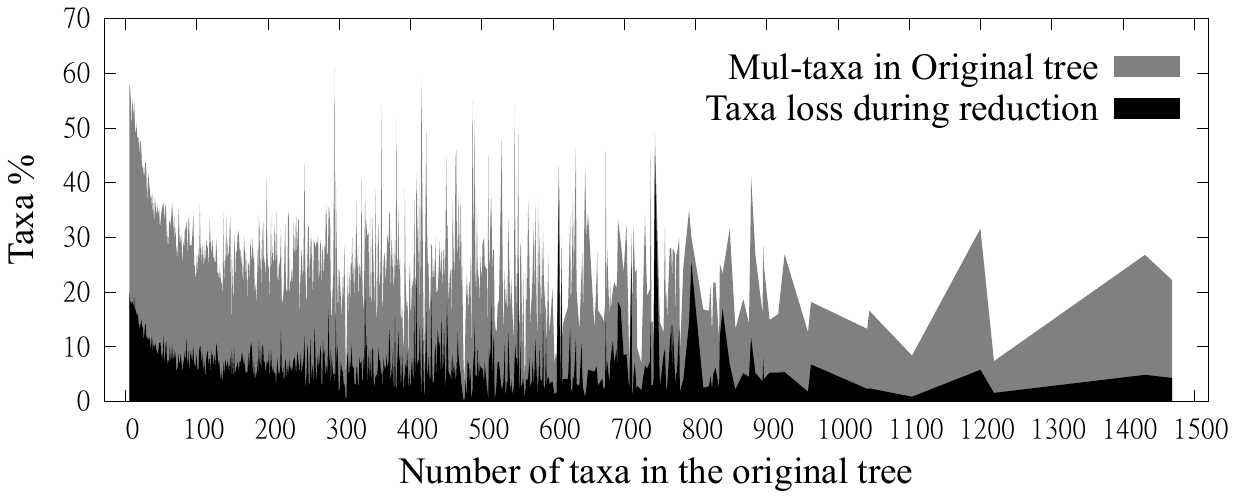}
\par\end{centering}

\label{Flo:TaxaLoss}}
\par\end{centering}

\begin{centering}
\subfloat[Quality of reduced singly-labeled trees]{\begin{centering}
\includegraphics[width=0.82\columnwidth]{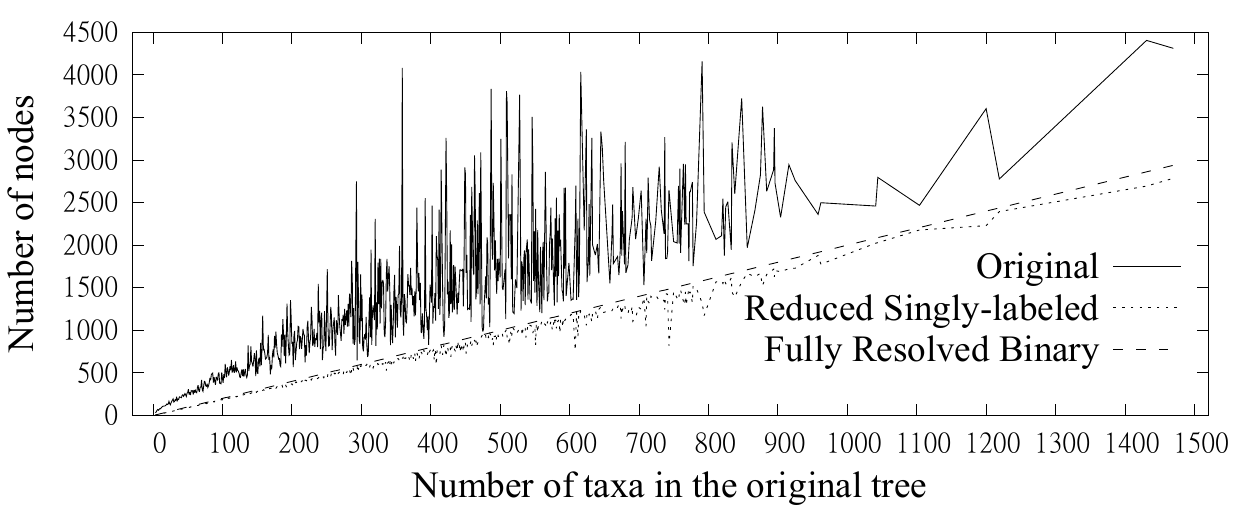}
\par\end{centering}

\label{Flo:reducedTreeQuality}}
\par\end{centering}

\caption{Experimental results}
\label{Flo:MulTreeExp}%
\end{figure}

We have integrated our reduction algorithm into SearchTree  (available at \url{http://searchtree.org/}),
a phylogenetic tree search engine that takes a user-provided list of species names and finds matches with a precomputed collection of phylogenetic trees, more than half of which are MUL-trees, assembled from GenBank sequence data.  The trees returned are ranked by a tree quality criterion that takes into account overlap with the query set, support values for the branches, and degree of resolution.   We have added functionality to provide reduced singly-labeled trees as well as the $\mulT$s based on the full leaf set. 

\section{\label{sec:Conclusion}Conclusion}

We introduced an efficient algorithm to reduce a multi-labeled $\mulT$
to a maximally reduced form with the same information content, defined as the set of non-conflicting quartets it resolves.  We also showed that the information content of a \mulT\ 
uniquely identifies the $\mulT$'s maximally reduced form.
This has potential application in comparing $\mulT$s by significantly
reducing the number of comparisons as well as in extracting species-level information efficiently and conservatively from large sets of trees, irrespective of the underlying cause of multiple labels.  Our algorithm can easily be adapted to work for rooted trees.

Further work investigating the relationship of the MRF to the original tree under various biological circumstances is also underway. We might expect, for example, that well-sampled nuclear gene families reduce to very small MRF trees, and that annotation errors in chloroplast gene sequences (in which we expect little gene duplication), result in relatively large MRF trees.  Comparing the MRF to the original MUL-tree may well provide a method for efficiently assessing and segregating data sets with respect to the causes of multiple labels.

It would be interesting to compare our results with some of the other approaches for reducing $\mulT$s to singly-labeled trees (e.g., \cite{scornavacca2010building}) or, indeed, to evaluate if our method can benefit from being used in conjunction with such approaches.

\section*{Acknowledgements}

We thank Mike Sanderson for helping to motivate this work, for many discussions about the problem formulation, and for our ongoing collaboration in the SearchTree project.  Sylvain Guillemot listened to numerous early versions of our proofs and offered many insightful comments. We also thank the anonymous reviewers for their comments, which helped to improve 
this paper.

\bibliographystyle{abbrv}
\bibliography{IEEEabrv,Mul-treeReduction}

\section*{Appendix: Proofs}

\subsection*{Proofs for Section~\ref{sec:Preliminaries-and-Definitions}}

\paragraph{Proof of Theorem~\ref{thm:singly_labeled}.}
Repeat the following step until $T$ has no multiply-occurring labels. 
Pick any multiply-occurring label $\ell$
in $T$, select an arbitrary leaf labeled by $\ell$, and relabel every other leaf  labeled by $\ell$, by a new, unique, label. 
The resulting tree $T'$ is singly labeled, and all labels of $T$ are also present in $T'$. Consider
a quartet $ab|cd$ in $T$, that is resolved by edge $(u,v)$.  Assume
that $\{a,b\}\in M_{u}^{uv}$ and $\{c,d\}\in M_{v}^{uv}$. Thus, $T_{u}^{uv}$ contains all the occurrences of label $a$. Clearly,
this also holds for the only occurrence of $a$ in $T'$. Similar
statements can be made about labels $b$, $c$, and $d$.
Thus, the quartet $ab|cd$ is resolved by edge $(u,v)$
in $T'$, and, hence, $T'$ displays all quartets of
$T$.  \inputencoding{latin1}{\hfill{}}\inputencoding{latin9}\qed

\paragraph{Proof of Lemma~\ref{lem:2}.}
Refer to Fig.~\ref{Flo:ProofLemma1}.  
Since $T_{u}^{uv}$ is a subtree of $T_{w}^{wx}$, $M_{u}^{uv}\subseteq M_{w}^{wx}$
by definition of $M_{u}^{uv}$. Thus, if $|M_{u}^{uv}|=|M_{w}^{wx}|$, we must have $M_{w}^{wx}=M_{u}^{uv}$ and, if $|M_{u}^{uv}| \neq |M_{w}^{wx}|$, we must have $M_{u}^{uv}\subset M_{w}^{wx}$.\hfill{} $\qed$

\paragraph{Proof of Lemma~\ref{lem:3}.}
\emph{(Only if)}
Suppose $\Delta(u,v)\subseteq\Delta(w,x)$; therefore, $M_{v}^{uv}\subseteq M_{x}^{wx}$.  By definition, $M_{v}^{uv}\supseteq M_{x}^{wx}$; hence, $M_{v}^{uv}=M_{x}^{wx}$. 

\emph{(If)} Suppose $M_{v}^{uv}=M_{x}^{wx}$.  By definition, $M_{u}^{uv}\subseteq M_{w}^{wx}$, which 
implies that $\Delta(u,v)\subseteq\Delta(w,x)$. \hfill{} $\qed$

\paragraph{Proof of Lemma~\ref{lem:4}.}
By Lemma \ref{lem:3}, since
$\Delta(u,v)\subseteq\Delta(w,x)$, we have $M_{v}^{uv}=M_{x}^{wx}$.
Now consider an edge $(y,z)$ on $P_{u,x}$. By definition $M_{v}^{uv}\supseteq M_{z}^{yz}\supseteq M_{x}^{wx}$.  But $M_{v}^{uv}=M_{x}^{wx}$, therefore $M_{v}^{uv}=M_{z}^{yz}=M_{x}^{wx}$.
By definition $M_{u}^{uv}\subseteq M_{y}^{yz}\subseteq M_{w}^{wx}$.
Hence, by Lemma \ref{lem:3}, $\Delta(u,v)\subseteq\Delta(y,z)\subseteq\Delta(w,x)$. \hfill$\qed$

\subsection*{Proofs for Section~\ref{sec:MRF}}

\paragraph{Proof of Theorem \ref{thm:5}.}
We rely on two facts.  First, every internal node in the tree has
degree at least three.  Second, every internal edge
in the tree resolves a quartet; otherwise, the edge would be contractible and the tree would not be maximally reduced. 

Consider any edge $(u,v)$ in the tree. To prove that $(u,v)$ resolves a 
quartet not resolved by any other edge, we need to show that there exists a quartet $ab|cd$ of the form 
shown in Fig.~\ref{Flo:uniqueQuartet}.
First, we describe how to select leaves
$a$ and $b$. Consider the following cases: 

\begin{case}
$u$ has at least two neighbors $i$ and $j$, apart from $v$, that
are internal nodes. Then, we select any $a\in M_{i}^{ui}$
and any $b\in M_{j}^{uj}$.
\end{case}

\begin{case}
$u$ has only one neighbor $i \neq v$ that is an internal
node. Then, at least one of  $u$'s neighboring leaves
must participate in a quartet that $(u,v)$ resolves. Without such
a leaf, $(u,v)$ would resolve the same set of quartets as $(u,i)$, so one of these two edges would be contractible, contradicting the assumption that the tree is maximally reduced. We select this leaf as $b$ and we select
any $a\in M_{i}^{ui}$.
\end{case}

\begin{case}
All neighbors of $u$, except $v$, are leaves. Then, at
least two of its neighbors must participate in a quartet, because $(u,v)$
must resolve a quartet. We select the two neighbors as $a$ and $b$. 
\end{case}
In every case, we can select the desired leaves $a$ and $b$. By a similar argument, we can also select the desired $c$ and $d$. This proves the existence of the desired quartet $ab|cd$. Therefore, each internal edge of $T$ uniquely
resolve a quartet. 
\hfill{}\qed

\paragraph{Proof of Theorem \ref{thm:10}.}
We need two lemmas.

\begin{lemma}
\label{lem:8}There is a bijection $\phi$ between the respective sets of internal edges of $T$ and $T'$  with the following property.  Let $(u,v)$ be an internal edge in $T$ and let $(u',v') = \phi(u,v)$.  Then, $M_u^{uv} = M_{u'}^{u'v'}$ and $M_v^{uv} = M_{v'}^{u'v'}$.  Therefore, $\Delta(u,v)=\Delta(u',v')$.
\end{lemma}

\begin{wrapfigure}{hr}{0.40\textwidth}
\vspace{-35pt}
\begin{center}

\includegraphics[width=0.35\columnwidth]{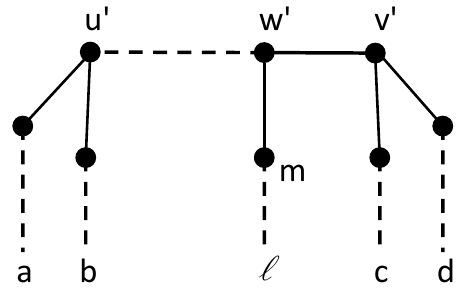}

\par\end{center}
\vspace{-20pt}
\caption{}
\vspace{-10pt}
\label{Flo:bijectionProof_2}

\vspace{-10pt}
\end{wrapfigure}

\noindent \emph{Proof.}  Consider an edge $(u,v)$ in $T$. By Theorem \ref{thm:5}, $(u,v)$ must resolve a quartet $ab|cd$ not resolved by any other edge as shown in Fig.~\ref{Flo:uniqueQuartet}. 
We claim that this quartet must be resolved uniquely by
an edge $(u',v')$ in $T'$. Suppose not. Using arguments similar to those in the proof of Lemma
\ref{lem:4}, we can show that all edges that resolve $ab|cd$ in $T'$ form a path $(u', x', \ldots, w', v')$, where possibly $x' = w'$, as shown in Fig.~\ref{Flo:bijectionProof_2}. Here,
$\{a,b\}\subseteq M_{u'}^{u' x'}$ and $\{c,d\}\subseteq M_{v'}^{w'v'}$.

Since $(w',v')$ resolves a quartet not resolved by any other edge, by Theorem \ref{thm:5}
there exists a label $\ell$ as shown, where $\ell \in M_{m}^{w'm}$.
Since $ab|\ell d$ is a quartet in $T'$ and $\Info(T)=\Info(T)$, it must be true that
$\ell \in M_{v}^{u v}$ in $T$. Clearly, $T$ does not resolve the quartet on 
$\left\{ a,\ell,d,c\right\} $ in the same way, $a\ell|cd$, as $T'$. This contradicts the assumption that $\Info(T)=\Info(T')$.
Thus, $(u',v')$ must be an edge. Moreover, only one such edge
exists in $T'$ as it uniquely resolves the quartet $ab|cd$.

Now consider any label $f\in M_{u}^{uv}$ such that $f\notin\left\{ a,b,c,d\right\} $.
Label $f$ must be in $M_{u'}^{u' v'}$; otherwise, $T$ and $T'$ would resolve the quartet $\left\{ a,f,c,d\right\}$ differently. Similarly, any such $f\in M_{u'}^{u' v'}$ must
be in $M_{u}^{uv}$ as well. Thus $M_{u}^{u v}=M_{u'}^{u' v'}$.
In the same way, we can prove that $M_{v}^{u v}=M_{v'}^{u' v'}$. Thus, $\Delta(u,v)=\Delta(u',v')$.

We have shown that there is a one-to-one mapping $\phi$ from edges in $T$ to edges of $T'$ such that $\Delta(e)=\Delta(\phi(e))$.  To complete the proof, we show that $\phi$ is onto.  Suppose that for some edge $e'$ in $T'$ there is no edge $e$ in $T$ such that $\phi(e) = e'$. But then $e'$ must resolve a quartet not resolved by any other edge in $T'$. This quartet cannot be in $\Info(T)$, contradicting the assumption that $\Info(T)=\Info(T')$.
\hfill{} $\qed$ \bigskip{}

Let $\phi$ be the bijection between the edge sets of $T$ and $T'$ from the preceding lemma.

\begin{lemma}
\label{lem:9} Let $(u,v)$ and $(v,x)$ be any two neighboring internal edges
in $T$, and let $(p,q) = \phi(u,v)$ and $(r,s) = \phi(v,x)$ be the corresponding edges in $T'$ such that $M_{u}^{u v}=M_{p}^{p q}$ and $M_{v}^{v x}=M_{r}^{r s}$.  Then,  $(p,q)$ and $(r,s)$ are neighbors in $T'$ with $q = r$.\end{lemma}

\begin{proof}
Since $(u,v)$ and $(v,x)$ are neighbors, and each resolves a
quartet that is not resolved by the other, $M_{u}^{uv}\subset M_{v}^{v x}$
and $M_{v}^{uv}\supset M_{x}^{v x}$. By Lemma \ref{lem:8}, this implies that $M_{p}^{pq}\subset M_{r}^{rs}$
and $M_{q}^{pq}\supset M_{s}^{rs}$. Thus,
the only way $(p,q)$ and $(r,s)$ can exist in $T'$
is as part of the path $P_{p,s} = (p,q, \dots , r, s)$. If $q\neq r$,
then consider the edge $(t,r)$ on $P_{ps}$ such that $p$ is closer to $t$ than to $r$.
Then, the following must hold:

\begin{equation}
M_{p}^{pq}\subset M_{t}^{tr}\subset M_{r}^{rs}\label{eq:lemma9-1}\end{equation}
and
\begin{equation}
M_{q}^{pq}\supset M_{r}^{tr}\supset M_{s}^{rs}\label{eq:lemman9-2}\end{equation}

Let $(z,w) = \phi^{-1}(t,r)$ be the edge in $T$ corresponding to $(t,r)$.
Irrespective of the position of $(z,w)$ in $T$, 
(\ref{eq:lemma9-1}) and (\ref{eq:lemman9-2}) cannot be simultaneously
true with respect to edges $(u,v)$, $(v,x)$ and $(z,w)$ in
$T$.  Therefore, $q = r$, which proves the desired result.
\hfill{} $\qed$
\end{proof}

Lemmas \ref{lem:8} and \ref{lem:9} show that $T$ and $T'$ are isomorphic with respect to their internal edges.  It remains to show a one-to-one correspondence between their leaf sets. For this, we match up the leaves attached at every pendant
node in $T$ and $T'$. We start with pendant nodes that have
only one internal edge attached to them. For example, consider an internal
edge $(u,v)$ in $T$ such that $v$ is a pendant node and $T_{v}^{u v}$
has only leaves. Let $(u',v') = \phi(u,v)$ be the corresponding edge
in $T'$ such that $M_{u}^{uv} = M_{u'}^{u v}$. By Lemma \ref{lem:8}, $C^{uv}=C^{u' v'}$. Moreover, neither $T$ nor $T'$
have prunable leaves.
Thus, the same set of leaves must be attached at $v$ and $v'$
respectively. In subsequent steps, we select an internal edge
$(u,v)$ in $T$ such that $v$ is a pendant node and all the
other pendant nodes in $T_{v}^{u v}$ have already been matched up
in previous iterations. Again, let $(u',v')= \phi(u,v)$ such that $M_{u}^{uv} = M_{u'}^{u v}$. Using similar arguments, the same set of
leaves must be attached at $v$ and $v'$ respectively. Proceeding
this way, each pendant node in $T$ can be paired with the
corresponding pendant node in $T'$, and be shown to have the
same set of leaves attached to them. This shows that $T$ and $T'$
are isomorphic, as claimed. 
\hfill{} $\qed$

\paragraph{Proof of Lemma~\ref{lm:prune_n_contract}.}
Refer to Fig.~\ref{Flo:ProofLemma4}.

\begin{figure}
\begin{centering}
\vspace{-10pt}
\par\end{centering}

\begin{centering}
\includegraphics[width=0.8\linewidth]{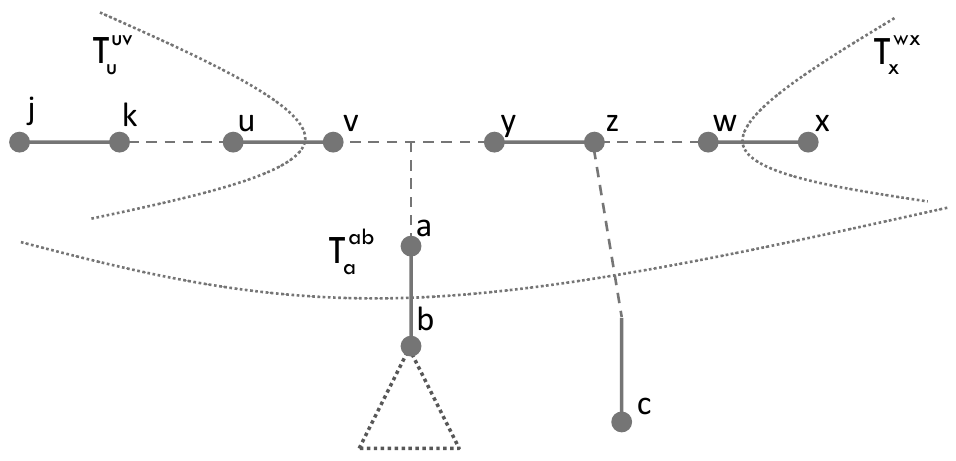}
\par\end{centering}

\vspace{-5pt}\caption{}

\vspace{-10pt}

\label{Flo:ProofLemma4}%
\end{figure}

\begin{enumerate}[(i)]
\item Consider any edge $(a,b)$ in a subtree branching out of $P_{u,x}$, as shown. We claim that  $M_{a}^{ab} \cup C^{ab} = M$; i.e., all the labels in $M$ appear in $T_{a}^{ab}$. This means that $M_{b}^{ab}=\emptyset$, so $(a,b)$ is uninformative.  

To prove the claim, observe first that, by definition, 
$M_{u}^{uv}\cup C^{uv}\cup M_{v}^{uv}=M.$
By Lemma~\ref{lem:3}, since $\Delta(u,v) \subseteq \Delta(w,x)$, we have $M_{x}^{wx} = M_{v}^{uv}$, so 
\begin{equation}
\label{eq:uninf}
M_{u}^{uv}\cup C^{uv}\cup M_{x}^{wx}=M.
\end{equation}
Now, $M_{u}^{uv}\cup C^{uv}$ is the set of labels on the leaves of $T_{u}^{uv}$, while
every label in $M_{x}^{wx}$ appears in $T_{x}^{wx}$.  Hence, $T_{u}^{uv}$ and $T_{x}^{wx}$ jointly contain every label in $M$.  Since $T_{u}^{uv}$ and $T_{x}^{wx}$ are subtrees of $T_{a}^{ab}$, this completes the proof of the claim.

\item Suppose $\Delta(u,v)=\Delta(w,x)$. By an argument similar to the one used in the
proof of Lemma~\ref{lem:4}, we can show that any edge
$(y,z)$ on the path $P_{v,w}=(v \ldots w)$ (see Fig.~\ref{Flo:ProofLemma4}) satisfies $M_{v}^{uv}=M_{z}^{yz}=M_{x}^{wx}$
and $M_{u}^{uv}=M_{y}^{yz}=M_{w}^{wx}$. Consider a leaf $c$
as shown; let $\ell$ be its label. Then, $\ell$ appears in $T_{x}^{wx}$, for else $M_{y}^{yz}\neq M_{w}^{wx}$, a contradiction.
Similarly, $\ell$ appears in $T_{u}^{uv}$. 
Now, let $S$ be the
tree obtained after pruning leaf $c$. 

\begin{enumerate}
\item $\Info(T)\subseteq\Info(S)$: Suppose pruning $c$ removes a quartet from $\Info(T)$. If
such a quartet exists in $T$, it must be resolved by an edge
$(j,k)\in T_{u}^{uv}$ (say). But then $(j,k)$ still resolves the
same quartet in $S$ because $\ell \in M_{x}^{wx}$, and the labels in $T_x^{wx}$ are 
a subset of those in $T_{k}^{jk}$.  This is a contradiction.
\item $\Info(S)\subseteq\Info(T)$: Suppose pruning $c$ adds a quartet to $\Info(S)$ that is not in $\Info(T)$. Such a quartet in $S$ must be resolved by an edge $(j,k)$ in $S_{u}^{uv}$
(say), that before pruning satisfied $\ell \in C^{jk}$, but now has $\ell \notin M_{k}^{jk}$. However $\ell \in M_{x}^{wx}$; therefore we still have $\ell \in C^{jk}$ and the edge still
cannot resolve the quartet, a contradiction.
\end{enumerate}
Hence, $c$ is prunable.\hfill{} $\qed$
\end{enumerate}

\paragraph{Proof of Lemma~\ref{lm:prunable}.}
Refer to Fig.~\ref{Flo:redundantCase}. Consider any pendant node $v$ of degree at least three in $T'$ attached to a leaf labeled $\ell$. Clearly deleting the leaf  does not change the information content of any edge in $T_{u}$ or $T_{y}$.
Now consider an edge $(w,x)$ in $T'$ as shown. Note that $\ell \in C^{wx}$, so $\ell$ does not contribute to $\Delta(w,x)$.  After deleting the leaf, we still have $\ell \in C^{wx}$, so $\Delta(w,x)$ remains unchanged.  Therefore, the leaf is prunable. \hfill{} $\qed$
\begin{figure}
\vspace{-10pt}

\begin{centering}
\includegraphics[width=0.8\columnwidth]{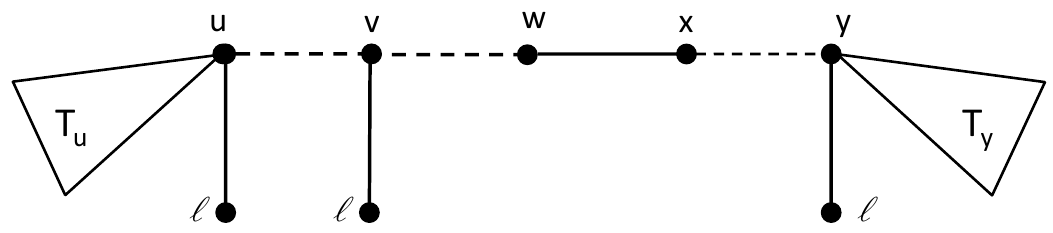}
\par\end{centering}

\vspace{-0pt}

\caption{Illustration for the proof of Lemma~\ref{lm:prunable}.  The leaves attached to pendant
nodes $u$, $v$, and $y$ are labeled by $\ell$, and the subtrees indicated
by $T_{u}$ and $T_{y}$ do not contain a leaf labeled with $\ell$.  Nodes $u$ and $y$ have degree two in $T'$, while $v$ has degree three.}

\vspace{-15pt}

\label{Flo:redundantCase}%
\end{figure}

\end{document}